\documentclass[a4paper,UKenglish,cleveref, autoref, thm-restate, nolineno]{socg-lipics-v2021}

\hideLIPIcs  

\usepackage{amsmath}
\usepackage{amsfonts}
\usepackage{amsthm}
\usepackage{amssymb}
\usepackage{dsfont}
\usepackage{enumitem}
\usepackage{graphicx}
\usepackage{faktor}

\usepackage{xcolor}
\usepackage{verbatim}

\newcommand{\Z}{\mathbb{Z}}					
\newcommand{\R}{\mathbb{R}}					
\newcommand{\N}{\mathbb{N}}					

\newcommand{\cpx}{\mathbb{K}}               
\newcommand{\sls}[1]{\cpx_{\leq #1}}        
\newcommand{\spx}{\sigma}                   

\newcommand{\mf}{M}                         

\newcommand{\wf}{w}                         
\newcommand{\gf}{\omega}                    
\newcommand{\bn}{b}                         

\newcommand{\C}{C}                          
\newcommand{\Zy}{Z}                         
\newcommand{\Hom}{H}                        

\newcommand{\ch}{c}                         
\newcommand{\zy}{\zeta}                     
\newcommand{\ozy}{z}                        
\newcommand{\hc}{h}                         

\newcommand{\bo}{\partial}                  
\newcommand{\Ld}{{\mathcal{L}}}             

\newcommand{\es}{\varnothing}				

\title{
    On Complexity of Computing Bottleneck and 
    Lexicographic Optimal Cycles in a Homology Class
}
\date{}
\titlerunning{On Complexity of Computing Bottleneck and Lexicographic Optimal Cycles} 

\author
    {Erin Wolf Chambers}
    {Saint Louis University, Saint Louis, MO, USA 
    \and \url{https://cs.slu.edu/~chambers/}}
    {erin.chambers@slu.edu}
    {}
    {This author was funded in part by the National Science Foundation 
    through grants CCF-1614562, CCF-1907612, CCF-2106672, and DBI-1759807.}

\author
    {Salman Parsa}
    {University of Utah, Salt Lake City, UT, USA}
    {sparsa@sci.utah.edu}
    {}
    {This author was funded in part by the Saint Louis University Research Institute and 
    by NSF grant CCF-1614562.}

\author
    {Hannah Schreiber}
    {Saint Louis University, Saint Louis, MO, USA}
    {hannah.schreiber.k@gmail.com}
    {https://orcid.org/0000-0002-8564-415X}
    {This author was funded in part by the National Science Foundation 
    through grant DBI-1759807.}

\authorrunning{E.\,W. Chambers and S. Parsa and H. Schreiber} 

\Copyright{Erin W. Chambers and Salman Parsa and Hannah Schreiber} 



\EventEditors{Xavier Goaoc and Michael Kerber}
\EventNoEds{2}
\EventLongTitle{38th International Symposium on Computational Geometry (SoCG 2022)}
\EventShortTitle{SoCG 2022}
\EventAcronym{SoCG}
\EventYear{2022}
\EventDate{June 7--10, 2022}
\EventLocation{Berlin, Germany}
\EventLogo{socg-logo}
\SeriesVolume{224}
\ArticleNo{XX}

\begin{document}
\maketitle


\begin{abstract}
    Homology features of spaces which appear in applications, for instance 3D
    meshes, are among the most important topological properties of these objects.
    Given a non-trivial cycle in a homology class, we consider the problem of
    computing a representative in that homology class which is optimal. We study 
    two measures of optimality, namely, the lexicographic order of cycles (the
    lex-optimal cycle) and the bottleneck norm (a bottleneck-optimal cycle). 
    We give a simple algorithm for computing the lex-optimal cycle for a 1-homology
    class in a closed orientable surface. In contrast to this, our main result is that, in the case 
    of 3-manifolds of size $n^2$ in the Euclidean 3-space, the problem of finding 
    a bottleneck optimal cycle cannot be solved more efficiently than solving a 
    system of linear equations with an $n \times n$ sparse matrix. From this 
    reduction, we deduce several hardness results. Most notably, we show that for 
    3-manifolds given as a subset of the 3-space of size $n^2$, persistent 
    homology computations are at least as hard as rank computation (for sparse matrices) while ordinary homology computations can be done in $O(n^2 \log n)$ time. This is 
    the first such distinction between these two computations. Moreover, it 
    follows that the same disparity exists between the height persistent homology 
    computation and general sub-level set persistent homology computation for 
    simplicial complexes in the 3-space.
\end{abstract}

\section{Introduction}

Topological features of a space are those features that remain invariant under 
continuous, invertible deformations of the space.
Homology groups are one of the most important topological features which, while 
not a complete invariant of shape, nevertheless are computationally feasible 
and capture important structure, in the following sense. Let $\cpx$ denote our 
space, which we will assume is a simplicial complex. For any dimension $d$, 
there is a homology group\footnote{
In this work, we will always use $\Z_2$ coefficients, so that the homology 
groups are also vector spaces.
} $H_d(\cpx)$ that captures the $d$-dimensional structure present. The zero 
dimensional group encodes the connected components of $\cpx$; the group 
$H_1(\cpx)$ contains information about the closed curves in $\cpx$ which can not 
be ``filled'' in the space (often described as handles); and the group 
$H_2(\cpx)$ captures the voids in the space that could not be filled, etc.\footnote{Note that this is a high level, intuitive description; we refer 
the reader to~\cite{munkres,hatcher} for more precise definitions.} 
For example, a hollow torus contains a single void and two classes 
of curves that are not ``filled'' in the space, and these features remain under 
continuous, invertible deformations of the shape.  

Although the above intuitive description of the homology features is useful in 
many applications, in general, homology groups are algebraic objects defined 
for a simplicial complex (or a topological space) which do not easily translate 
to a canonical geometric feature. An element of a $d$-dimensional homology group 
is a homology class, and a homology class by definition contains a set of 
$d$-cycles, where cycles which are in the same class are called homologous to 
each other. Assume our complex $\cpx$ is a 3D mesh and $d = 1$. A cycle under 
homology is a set of edges in the mesh, such that each vertex is incident to an 
even number of edges. A fixed cycle therefore corresponds to a fixed geometric 
feature of the mesh, while the homology class contains a large collection of 
these cycles. Cycles in the same class could be very different geometrically, 
see figure~\ref{fig:homologous}. Consequently, the knowledge of homology groups 
or Betti numbers (which are the dimensions of the homology groups) does not 
directly provide us with geometric features that lend themselves to 
interpretations that are necessary for many applications, especially in 
topological data analysis. 
\begin{figure}
    \centering
    \includegraphics[scale=0.6]{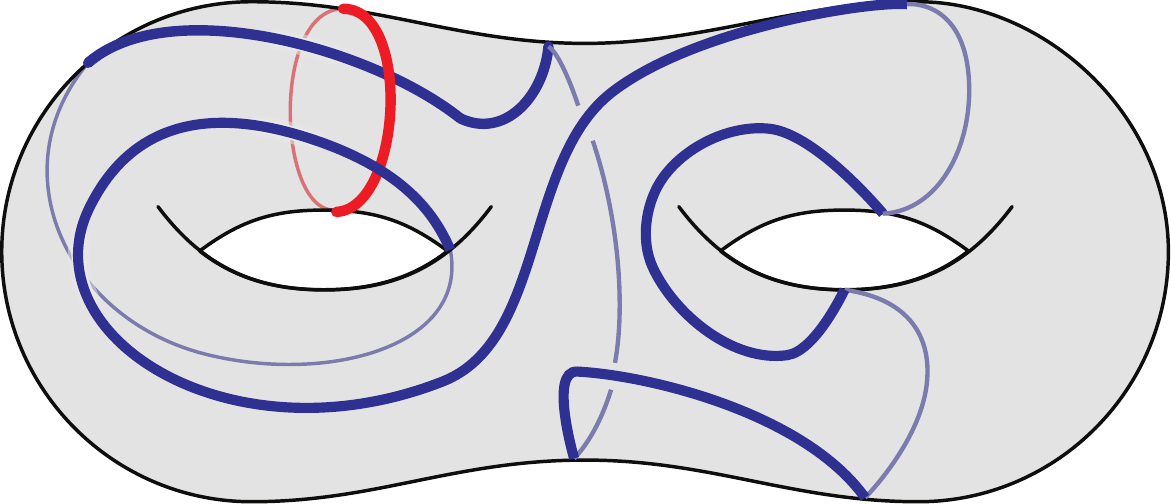}
    \caption{The blue cycle is homologous (with $\Z_2$ coefficients) to 
    the red cycle.}
    \label{fig:homologous}
\end{figure}
Therefore, it is desirable to assign unique cycles, or those with known geometric 
features, to a homology class in a natural way. Much recent work has sought to 
define measures or weights on the cycles and then represent each homology class 
by some (ideally unique) cycle which optimizes that measure. This problem has 
been well-studied in the literature with many different measures proposed; see 
Section~\ref{ssec:relatedwork} for an overview of relevant results. 
Interestingly, sometimes optimizing the cycle is NP-hard and sometimes polynomial-time, depending on the measure, the classes of spaces one allows in the input, 
and the type of homology calculation. One of the more widely studied versions 
gives each edge a weight and then seeks the representative with minimum total 
length in the homology class. This problem is known as 
\emph{homology localization}, and even its complexity varies widely depending on 
the space and the type of homology calculation.
There is also a rich body of work that seeks to compute optimal cycles in 
persistent homology classes; again, we refer to Section~\ref{ssec:relatedwork} 
for details and citations.

In most of the paper we fix a simplicial complex $\cpx$, a fixed $d$, and a weight function given on $d$-simplices of $\cpx$. However, unlike traditional homology localization, 
for most of our work, it suffices to think of the weight function as an ordering 
of the simplices. We then consider two measures that this ordering induces on 
the set of $d$-cycles. The first, already defined and studied in Cohen-Steiner 
et al.~\cite{cohensteiner2020}, is the lexicographic ordering on the chains. The 
second is a minmax measure we call the bottleneck norm, which assigns to each 
chain the maximum weight of a simplex in it. We note that computing the 
lexicographic-optimal cycle is at least as hard as computing a bottleneck-optimal 
cycle in a given homology class, as the lexicographic-optimal cycle is 
always bottleneck-optimal (but the reverse is not always true). In the rest of the paper, we often shorten lexicographical-optimal to lex-optimal.

\subsection{Contributions}

It is proved in~\cite{cohensteiner2020} that the persistent homology boundary 
matrix reduction can be used to compute the lex-optimal cycle in any given 
homology class in cubic time in the size of the complex, for any dimension. 
In this paper, we begin by presenting a new simple algorithm that, given a 
(closed orientable) surface and a 1-dimensional cycle, computes a lex-optimal cycle homologous to 
the input cycle in $O(m \log m)$ time, where $m$ is the size of a triangulation 
of the surface. We note that an algorithm with slightly better running-time ($O(n \alpha(n))$, where $\alpha(n)$ is the inverse Ackermann function) is also 
given in~\cite{cohensteiner2020} although their algorithm only works for cycles which are homologous to a boundary and satisfy some other restrictions, see \cite[Problem 17]{cohensteiner2020}.



The simplest setting after surfaces is perhaps 3-manifolds embedded in Euclidean 
3-space, for instance solid 3D meshes. For simplifying run-time comparisons, we 
denote the sizes of the complexes in $\R^3$ by $n^2$. Our main contribution in 
this paper is that given a system $Ax=b$ of linear equations with $A$ a 
sparse\footnote{
By a sparse $n\times n$ matrix we mean that the number of non-zeros is at most 
$cn$ for some constant $c$.
} 0-1 matrix, it is possible to construct in $O(n^2 \log n)$ time a 3-manifold 
embedded in $\R^3$ of size $O(n^2)$ such that solving the system for a solution 
$x$ is equivalent to computing a bottleneck-optimal cycle in a given homology 
class. Our reductions remain true for integer homology and other fields $\Z_p$ 
(with an appropriate definition of optimal cycles), albeit with slight changes 
in the run-time of the reductions.


In~\cite{Dey2019}, Dey presents an algorithm for computing the persistent 
diagram of a height function for a complex in $\R^3$ (of size $n^2$) in 
$O(n^2 \log{n})$ time. In addition, in the same running time a set of 
generators can be computed. From our reduction, it follows that, if the given 
function on the complex (which is a mesh in $\R^3$) is not a height function, 
then these computations cannot be done faster than rank computation 
for a sparse 0-1 matrix. This gives a first answer to the main question asked 
in~\cite{Dey2019}, asking if efficient algorithms exist for the non-height 
functions. In other words, our results show that there is a disparity between the 
efficiency of algorithms for computing sub-level-set persistence for 3D meshes 
of height and of general functions.

Ordinary Betti numbers for complexes in $\R^3$ (of size $n^2$) can be computed 
in $O(n^2 \log n)$ time~\cite{DeEd95} (if a triangulation of the complement is also given). It follows from our reduction that computing persistence Betti numbers for an arbitrary function for complexes in $\R^3$ is as hard as computing the rank of a sparse 0-1 $n \times n$ matrix (even if a triangulation of the complement is given). 
To our knowledge, this is the first such distinction between persistent and 
ordinary homology computations.

We should also mention that the significance of the reductions, like the ones 
presented in Section~\ref{s:reductions}, is not giving a lower bound for the problem in the 
complexity theoretic sense, as we have not done this, since we do not know if 
solving a sparse system has a non-trivial lower bound. Rather, the reductions 
show that the geometry of the problem does not help in improving trivial 
deterministic algorithms. For instance, one of our theorems tells a researcher
of geometric methods that it is futile to try to find a deterministic 
$O(m\log{m})$ algorithm that computes persistent Betti numbers for meshes in 
3D if that researcher is not interested in improving the best run-time for 
matrix rank computation for sparse matrices. As mentioned before, an 
$O(m \log(m))$ algorithm exists if we are interested only in height functions. 
Here, $m$ is size of the input mesh.



%
\subsection{Related work}
\label{ssec:relatedwork}

The Optimal Homologous Chain Problem (OHCP) is a well studied problem in 
computational topology, which specifies a particular cycle or homology class and 
asks for the ``optimal'' cycle in the same homology class. Similarly, the problem 
of homology localization~\cite{Zomorodian2008} specifies a topological ``feature'' 
(usually a homology class, such as a handle or void), and asks for a 
representative of that class. Such representatives can be used for simplification, 
mesh parametrization, surface mapping, and many other problems.

Of course, computability and practicality often depend on the exact definition 
of ``optimal'', with a wide range of variants. One natural notion of optimal is 
simply to assume the input complex has weights on the simplices, and to compute 
the representative of minimum length, area, or volume (depending on the 
dimension).  Here, length (or area or volume) of a chain is computed as a 
weighted summation of the weights of its simplices; it then remains to specify 
the coefficients used when computing these objects, since the choice of 
coefficient can greatly affect the results. The resulting trade-offs can be 
quite subtle and surprising. For example, minimum length homologous cycles with 
$\Z$ coefficients in the homology class of highest dimension, if homology is torsion free, reduces 
to linear programming and hence is solvable in polynomial 
time~\cite{Dey2011,Chambers2014}.
In contrast, with $\Z_2$ coefficients the problem is NP-hard to compute, even on 
2-manifolds~\cite{Chambers2009,chambers2019minimum}. In fact, homology localization is NP-hard to 
approximate for $\Z_2$ coefficients within a constant 
factor even when the Betti number is constant~\cite{Chen2011}, and APX-Hard but 
fixed parameter tractable with respect to the size of the optimal 
cycle~\cite{Borradaile2020, BuCa-etal2012}.
When coefficients are over $\Z_k$, the problem becomes Unique Games Conjecture 
hard to approximate~\cite{Grochow2018}. Homology localization has also been 
studied under the lens of parameterized complexity, where it is fixed parameter 
tractable in treewidth of the underlying complex~\cite{Blaser2020}.

There has been considerable followup work on different variants of homology 
localization. One major line of work focuses on persistent homology generators, 
which are often related to homology localization but seek generators in a 
filtration which realize a particular persistent homology 
class~\cite{Chen2010,Oleksiy2010,Hiraoka2016,Obayashi2018,Dey2019,10.5555/3381089.3381247,dey2018persistent}; 
again, there is high variance on notions of optimality for these generators and 
on input assumptions, both of which affect complexity.
More directly related to this paper, as noted in the introduction, lexicographic 
minimum cycles under some ordering on the simplices have also been 
studied~\cite{cohensteiner2020}. 

Hardness of computing ordinary homology for complexes in Euclidean spaces is 
discussed in~\cite{EdPa14}, where a reduction to rank computation of sparse 
matrices is presented; the results of this paper thus in a sense extend those 
of~\cite{EdPa14}.

We note that there are randomized and probabilistic algorithms for sparse matrix 
operations in almost quadratic time~\cite{Wied1986, ChKw-wtal2013}. As a result,  
our reductions do not apply for these types of algorithms, since they take 
$O(n^2 \log(n))$ time for a matrix with $O(n)$ non-zeros. It is natural to ask 
for a reduction that is linear in the size of the input $A$; indeed, this 
presents an interesting direction of future research.


 \section{Background }

 We begin with a brief overview of terminology and background; for more detailed 
 coverage on these topics, we refer the reader to textbooks on 
 topology~\cite{munkres} and computational 
 topology~\cite{EdHa2010,Oudot2015,DeyWang2021}.

 \subparagraph{Simplicial complex}
 Let $V$ be a finite set. A (abstract) \emph{simplicial complex} $\cpx$ is a set 
 of subsets of $V$ such that if $\sigma \in \cpx$ and $\tau \subseteq \sigma$, 
 then $\tau \in \cpx$. An element of $\cpx$ is called an abstract simplex. If 
 $\sigma \in \cpx$ has $d+1$ elements then it is $d$-dimensional. The dimension 
 of a simplicial complex is the maximum dimension of its simplices. 
 A 0-dimensional simplex is called a \emph{vertex}, a 1-dimensional simplex an 
 \emph{edge}, 2-dimensional simplex a \emph{triangle}, 3-dimensional simplex a 
 \emph{tetrahedron}. The \emph{size} or the \emph{complexity} of a complex is 
 the number of its simplices.

 A simplex $\tau$ is called a proper face of a simplex $\sigma$ if 
 $\tau \subset \sigma$. For a simplex $\tau \in \cpx$, the \emph{star} of $\tau$ 
 is the set of simplices $\sigma \in \cpx$ such that $\tau \subset \sigma$ and 
 the \emph{closed star}, $S(\tau)$, additionally also includes all faces of the 
 simplices in the star. Therefore, the closed star is a simplicial complex, 
 whereas the star is not. The \emph{link} of a simplex $\tau$ is defined as 
 $L(\tau) = \{\sigma-\tau, \sigma \in S(\tau)\}$. The link of a simplex is also 
 a simplicial complex.

 A geometric \emph{$d$-simplex} is a convex hull of $d+1$ affinely independent 
 points $\{p_0, \ldots, p_{d}\}$ in $\R^d$. Affinely independent means that the set 
 $\{p_i - p_0, 1 \leq i \leq d \}$ is an independent set of vectors. 
 We say any geometric $d$-simplex realizes an abstract $d$-simplex. Note that a 
 face of a $d$-simplex is realized in the boundary of the geometric $d$-simplex. 
 A simplicial complex $\cpx$ uniquely determines a topological space called its 
 \emph{underlying space} and denoted $|\cpx|$. If the complex has $n$ vertices, 
 then all simplices of $\cpx$ are realized simultaneously in the boundary of an 
 $(n-1)$-simplex $\Delta$. We can define $|\cpx|$ to be the union of the geometric 
 realizations of the simplices in $\cpx$ on the boundary of $\Delta$.

 The $k$-skeleton of a simplicial complex $\cpx$ is the simplicial complex which is 
 the set of simplices of $\cpx$ of dimension at most $k$. Moreover, we note by 
 $\cpx_d$ the set of $d$-simplices of $\cpx$.

 A \emph{simplex-wise linear} function $f: |\cpx| \rightarrow \R$ is a function 
 determined by the value which it assigns to any vertex. In the relative interior 
 of any geometric simplex the function linearly interpolates between the vertex 
 values.
 Alternatively, $f$ is simplex-wise linear if the restriction of it to any 
 geometric simplex of $|\cpx|$ is linear. We call such a function \emph{generic} 
 if the values of vertices are distinct.

 We say that a simplicial complex $\cpx$ is realized in $\R^3$ or is given as a 
 subset of $\R^3$ if there a simplex-wise linear function 
 $f: |\cpx| \rightarrow \R^3$ which is one-to-one. Note that such a function is 
 uniquely determined by giving the three coordinates of any vertex of $\cpx$ in 
 $\R^3$. So, 3D meshes used in visualisation, computer graphics, etc. 
 are simplicial complexes realized in $\R^3$. 

 \subparagraph{Manifold} 
 A (topological) \emph{$d$-manifold} $\mf$ is a nice enough (i.e. second 
 countable, compact, and Hausdorff) topological space with the property that 
 every point in the space has a neighborhood that is homeomorphic to a 
 $d$-dimensional Euclidean space or a half-space.
 The boundary of the manifold, $\bo \mf$, is the set of points of the second 
 type.

 In this paper, we work with manifolds which are underlying spaces of simplicial 
 complexes. Such a simplicial complex is called a \emph{combinatorial manifold}. 
 For a combinatorial manifold we write $|\cpx| = \mf$.
 Intuitively, for a 2-manifold, we thus have an embedding of a the complex $\cpx$ 
 on a some underlying surface, such that vertices are mapped to distinct points 
 and edges are mapped to non-crossing curves. More precisely, in a 2-dimensional 
 combinatorial manifold, $\cpx$ is mapped onto the underlying surface such that 
 the link of a vertex is a 1-dimensional complex, i.e., 
 a graph, and the link of an edge is a set of vertices. 
 A simplicial complex is a combinatorial 2-manifold if and only if:
 \begin{itemize}
     \item the link of every vertex is either a simple circle or a simple path,
     \item the link of an edge is a pair of vertices or a single vertex.
 \end{itemize}
 Those vertices and edges whose links are paths and a single vertex define a 
 subcomplex which is a combinatorial 1-manifold, whose underlying space is the 
 boundary of the manifold $\mf$.

 A \emph{system of loops} on a 2-manifold $\mf$ is a set of pairwise disjoint 
 simple loops $L$ with a common base point such that $\mf \setminus L$ is a 
 topological disk. 
 On an orientable 2-manifold, any system of loops contains exactly $2g$ loops, $g$ being 
 the genus of $\mf$, and $\mf \setminus L$ is a $4g$-gon where each loop appears 
 as two boundary edges; this $4g$-gon is called the \emph{polygonal schema} 
 associated with $L$.
 We define a \emph{cut graph} of the complex $\cpx$ as a subgraph $G$ of the 
 1-skeleton of $\cpx$ such that $\cpx \setminus G$ is homeomorphic to a disk. 
 Cut graphs are generalizations of the cut locus, which is essentially the 
 geodesic medial axis of a single point. As we note above, every system of 
 loops trivially forms a cut graph, as its removal generates a polygonal schema, 
 but there may be many different cut graph in general.

 Algorithmic approaches on combinatorial 2-manifolds are often approached via a 
 tool called the \emph{tree-cotree decomposition}~\cite{eppstein2003}, which 
 is a partition of the edges of the 2-complex $\cpx$ into three sets, 
 $T \cup Q^* \cup L$, where $T$ is a spanning tree of the graph, $Q^*$ are the 
 dual edges of a spanning tree $Q$ of the dual graph, and $L$ is the set of 
 leftover edges of $\cpx$. Here, a \emph{spanning tree} is a tree formed by a 
 subset of the edges in $\cpx$ such that all vertices of $\cpx$ are included. 
 A \emph{cotree} is a tree in the dual graph. The 
 \emph{dual graph} of $\cpx$ has a vertex for each triangle in $\cpx$ and two 
 vertices are joined by an edges if the two corresponding triangles share an 
 edge. Euler’s formula implies that $|L| = 2g$.  In fact, each $\ell \in L$ 
 creates a cycle when added to $T$, and the collection of such cycles forms 
 a system of loops, which in turn generates a polygonal schema.

 A combinatorial 3-manifold $\mf$ is a 3-dimensional simplicial complex such that:
 \begin{itemize}
     \item the link of every vertex is a combinatorial 2-sphere or a 2-ball,
     \item the link of every edge is a circle or a half-circle,
     \item the link of every triangle is a pair of vertices or a single vertex.
 \end{itemize}
 Again the simplices of the second type define a combinatorial 2-manifold which is 
 the boundary of $\mf$.

 \subparagraph{Homology}
 Consider a simplicial complex $\cpx$. A \emph{$d$-chain} over the coefficient 
 field $\Z_2$ is a formal sum of $d$-simplices: $\sum_i t_i \spx_i$, where 
 $t_i \in \Z_2$ and $\spx_i \in \cpx_d$. The set of $d$-chains $\C_d(\cpx)$ is 
 called the $d$-dimensional \emph{chain group}, where we add the chains by adding 
 coefficients. 
 A chain can also be viewed as a set of simplices, where the simplex $\spx_i$ is 
 in the set whenever $t_i = 1$. Under this view, addition of $d$-chains is the 
 same as taking the symmetric difference of the simplices in each chain. 

 The \emph{boundary operator} $\bo_d: \C_d \rightarrow \C_{d-1}$ is a linear 
 transformation that to each $d$-simplex $\spx$ assigns the set of 
 $d-1$-simplices on its boundary. This defines $\bo_d$ uniquely on all the chains 
 by linearity. A chain is called a \emph{cycle} if its boundary is zero. The 
 $d$-cycles also form a vector space which we denote by 
 $Z_d(\cpx)\subset C_d(\cpx)$. A $d$-chain is called a \emph{boundary} if it lies 
 in the image of $\bo_{d+1}$. We denote the $d$-boundaries by $B_d(\cpx)$. 
 The most important property of the boundary operator is that for all $d\geq 0$, 
 $\bo_{d-1}\bo_{d}=0$. This relation implies that $B_d(\cpx) \subset Z_d(\cpx)$. 
 Therefore we can define the quotient 
 \[
     H_d(\cpx) = \faktor{Z_d(\cpx)}{B_d(\cpx)}.
 \]
 $H_d(\cpx)$ is called the $d$-dimensional \emph{homology group} of $\cpx$. It is 
 also a vector space and its dimension is called the $d$-th ($\Z_2$) Betti number, 
 denoted $\beta_d(\cpx)$.
 Observe that homology classes partition the cycles. If $z \in Z_d(\cpx)$ is a 
 cycle we denote by $[z] \in H_d(\cpx)$ its homology class. Two chains $c$ and 
 $c'$ are called \emph{homologous} if $c + c'$ is a boundary chain. Homologous 
 cycles belong to the same homology class.

 \subparagraph{Knots and Links}
 A \emph{knot} is a simple closed curve in $\R^3$, and a \emph{link} is a set of 
 disjoint knots in $\R^3$. Such objects are often represented and studied via 
 \emph{link diagrams}, or projections of the link to $\R^2$ which are injective 
 except for finitely many crossings, labeled to indicate which strand is crossing 
 over the other. See Figure~\ref{fig:knotdiagram} for a link example.

 \begin{figure}
     \centering
     \includegraphics[height=2in]{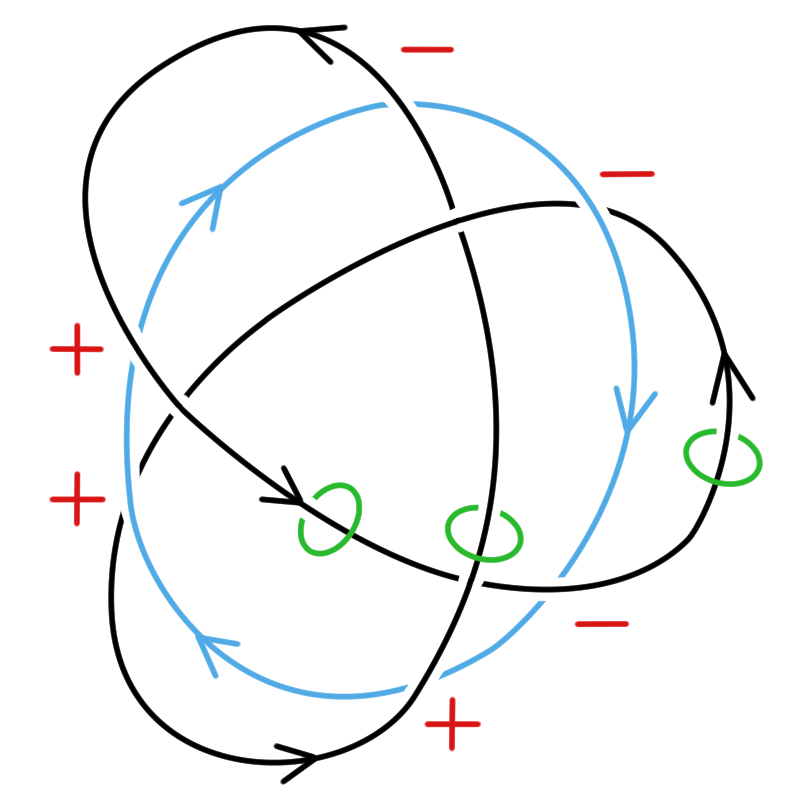}
     \caption{An example of a link diagram, with the crossing number values 
     shown in red for each crossing; the crossing number of this knot is 0. 
     Three meridians of the same link component of the diagram are also 
     indicated in green.}
     \label{fig:knotdiagram}
 \end{figure}

 Given a link (or knot) $L$, a \emph{tubular neighborhood} is an embedding 
 $t : L \times D^2 \rightarrow S^3$ so that $t(x,0) = x$ for $x \in L$ and $D^2$ 
 is an open unit disk; intuitively, this is simply a small thickening of the link 
 that does not introduce any intersections between strands. The 
 \emph{knot complement} of $L$ is $S^3$ minus the tubular neighborhood. The 
 \emph{meridian} of a spatial knot is a small circle that goes around the knot. 
Note that regardless of where we draw meridians, they are all homologous to each 
 other in the complement of the knot; again, see Figure~\ref{fig:knotdiagram}.

 The \emph{integer linking number} is an invariant to quantify the linking between 
 two closed knots, which intuitively measures the number of times that the curves 
 wind around each other. It is not a complete invariant: Any two unlinked curves 
 have linking number zero, but two curves with linking number zero may still be 
 linked. These can be formalized in several different ways; see 
 e.g.~\cite{RICCA2011}, although we use a simple combinatorial characterization 
 here: we label each crossing of the diagram as positive or negative, according 
 to the classification shown in Figure~\ref{fig:crossingnumber}.
 Then, the total number of positive crossings minus the total number of negative 
 crossings is equal to twice the linking number of the diagram. This defines the 
 integer linking number. The \emph{$\Z_2$ linking number} is just the parity of 
 the absolute value of the integer linking number, so it is either 0 or 1.

 By a \emph{spatial knot} or a \emph{spatial link}, we mean a simple closed curve 
 or a collection of disjoint simple closed curves respectively, in some standard 
 3-ball. The $\Z_2$ linking number between two spatial knots can simply be defined 
 as follows. Take any singular disk bounded by one curve and count the intersections of the 
 disk with the other curve. The parity of this number is the $\Z_2$ linking number.


 \begin{figure}
     \centering
     \includegraphics[width=4in]{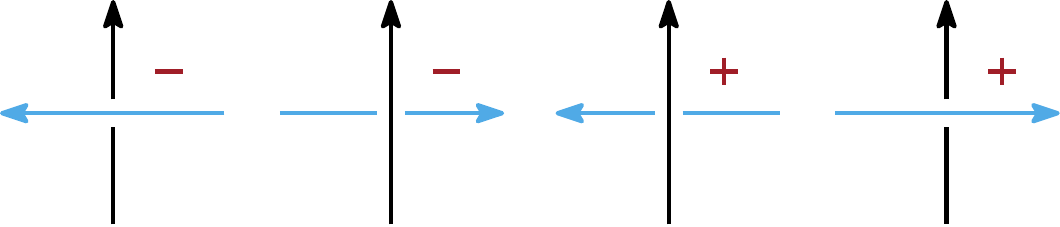}
     \caption{To compute the crossing number, each over/under crossing is labeled 
     as shown; the crossing number is the sum of all crossing labels divided by 
     two.}
     \label{fig:crossingnumber}
 \end{figure}

\section{Bottleneck and lex-optimal cycles}

Let $\cpx$ be a simplicial complex and $\cpx_d = \{\spx_0, \ldots, \spx_m\}$ be 
the set of $d$-dimensional simplices of $\cpx$. A \emph{weight function} $\wf$ 
on $\cpx_d$ is an arbitrary function 
$\wf: \cpx_d \rightarrow \R_{>0} = \{r \in \R \mid r>0 \}$. Thus $\wf$ is
defined on the generators of the chain group $\C_d(\cpx)$. For simplicity, we assume that $\wf$ is injective, i.e., simplices have distinct 
weights. For our purposes, such a weight function is equivalent to one with 
co-domain $\N - \{0\}$, or a total ordering of the simplices. If the weight 
function is not injective, then the edges with the same weight have exactly the 
same potential to appear in the optimal cycle and adding some small 
perturbations to their weights to distinguish them will not affect the 
consistency of the end result.

We extend $\wf$ to the function $\bn_{\wf}: \C_d(\cpx) \rightarrow \N$ as 
follows: for a chain $\ch \in \C_d(\cpx)$ of the form 
$\ch = \sum_{i=0}^{m} t_i \spx_i$, where, $\forall i$, $t_i \in \Z_2$, we set
\begin{equation*}
    \bn_{\wf}(\ch) =
    \begin{cases}
        \max_{
            \substack{0 \leqslant i \leqslant m \\ t_i = 1}
        }\{\wf(\spx_i)\} & \text{if $\ch \neq 0$}\\
        0 & \text{if $\ch = 0$}.
    \end{cases}       
\end{equation*}
In other words, if we view a chain $\ch \in \C_d(\cpx)$ as a set of simplices, 
$\bn_{\wf}$ assigns to $\ch$ the maximum weight of a simplex in $\ch$. We call 
$\bn_{\wf}$ the \emph{bottleneck norm} on $\C_d(\cpx)$. 

By the \emph{maximum simplex}, we mean the simplex with the largest 
weight in the chain. 

Although $C_d(\cpx)$ is a finite vector space, the function $\bn_{\wf}$ has 
properties analogous to a norm. First, it is non-negative. Second, assume $x$ 
and $y$ are chains and $\sigma_x$ and $\sigma_y$ are their maximum simplices. 
The maximum simplex of $x+y$ has weight at most 
$\max\{\wf(\sigma_x),\wf(\sigma_y)\} \leq \wf(\sigma_x) + \wf(\sigma_y)$. 
Hence $\bn_{\wf}$ satisfies the triangle inequality. And third, clearly if 
$\bn_{\wf}(\ch) = 0$ then $\ch = 0$. 

One can also define a lexicographic ordering on the $d$-chains based on the 
given weight function $\wf$, see also~\cite{cohensteiner2020}. 
For this purpose, we order the $d$-simplices such that $\sigma < \sigma'$ if 
and only if $\wf(\sigma) < \wf(\sigma')$. We assume that the subscript of the 
$\sigma_i$ respects the order. Let $\ch = \sum t_j \sigma_j$ and 
$\ch' = \sum t'_j \sigma_j$. We define $\ch <_L \ch'$ if there exists an index 
$j_0$ such that for $j > j_0$, $t_j = 1$ if and only if $t'_j = 1$, and 
$t_{j_0} = 0$, $t'_{j_0} = 1$. We write $\ch \leq_L \ch'$ if $\ch <_L \ch'$ or 
$\ch = \ch'$.

\subsection{Problem definitions}

In this section, we give formal definitions for our two main problems, the 
\emph{Bottleneck-Optimal Homologous Cycle Problem (Bottleneck-OHCP)} and the 
\emph{Lexicographic-Optimal Homologous Cycle Problem} 
(Lex-OHCP)~\cite{cohensteiner2020}, as well as defining optimal bases for 
homology groups.

\subparagraph{Bottleneck-OHCP} 
Given a weight function $\wf$ on $\cpx_d$, and a cycle $\zy \in \Zy_d(\cpx)$, 
compute a cycle $\ozy_*$ such that $[\ozy_*]=[\zy]$ and such that $\ozy_*$ 
minimizes the bottleneck norm. More formally, find $\ozy_*$ such that 
$\bn_{\wf}(\ozy_*) = \min 
\{\bn_{\wf}(\ozy) \mid 
\ozy \in\Zy_d(\cpx), \exists c \in \C_{d+1}(\cpx), \ozy + \zy = \bo c\}$.

In other words, the weight of the maximum simplex in $\ozy_*$ is 
minimized in the homology class of $\zy$. Therefore, we can also define the 
bottleneck weight function $\bn^*_{\wf}: \Hom_d(\cpx) \rightarrow \R_{\geq 0}$ 
on the homology classes by using the minimum $[\zy] \mapsto \bn_{\wf}(\ozy_*)$.
Thus the problem can also be formulated as computing the cycle which achieves 
$\bn^*_{\wf}(\hc)$ given any representative of the homology class $\hc$.

\subparagraph{Lex-OHCP} 
Given a weight function $\wf$ on $\cpx_d$, and a cycle $\zy \in \Zy_d(\cpx)$, 
compute the cycle $z_*$ such that $[z_*]=[\zy]$ and for any $d$-cycle $y$, 
if $[y] = [\zy]$ then $z_* \leq_L y$.

We note that by our convention on the weight function, the lex-optimal cycle is 
always unique. Moreover, the lex-optimal cycle is also bottleneck-optimal, 
however, the converse is not true.
Our reductions and hardness results are formulated for the bottleneck norm. 
Counter-intuitively, considering this intermediate problem simplifies our 
reductions and hardness proofs. 

\subparagraph{Optimal basis} 
For any suitable measure or weight function on the cycles we can define the 
corresponding optimal basis. Let $\leq_p$ be some pre-order on the set of 
$d$-cycles $Z_d(\cpx)$ such that every subset $A \subset Z_d(\cpx)$ has some 
chain $a$, such that $\forall z \in A, a \leq_p z$.

With respect to this pre-order, we define the \emph{optimal basis} for 
$d$-homology, as a set of cycles $B \subset Z_d(\cpx)$, representing the 
homology classes generating $H_d(\cpx)$, as follows. Put the smallest non-zero 
element of $Z_d(\cpx)$ in $B$. Now, repeat the following until $B$ is a 
representative basis for $d$-homology: let $A$ be the union of the cycles in 
the classes that are not in the subspace generated by the classes represented in $B$. Put the 
smallest cycle of $A$ in $B$.

In Section~\ref{s:algorithm}, we will describe a simple algorithm for computing 
the lex-optimal basis for the 1-dimensional homology of a surface.

\subsection{The Sub-level bottleneck weight function}

We defined the bottleneck weight function 
$\bn^*_{\wf} : \Hom_d(\cpx) \rightarrow \N$ on homology classes using a weight 
function on $d$-simplices for some fixed dimension $d$. Here we give a second, more natural definition of a 
generalization of this weight function. 
Let $\gf: |\cpx| \rightarrow \R$ be a generic simplex-wise linear function. The 
sub-level set of a value $r \in \R$ is the set 
$|\cpx|_{\leq r} = \{x \in |\cpx| \mid \gf(x) \leq r\}$. For any $d$-cycle 
$\zeta \in Z_d(\cpx)$, define 
$
    \bn_\gf(\zeta) := \min \{r \in \R \mid 
        \exists z \in Z^s_{d}(\sls{r}),\; 
        \exists y \in C^s_{d + 1}(\cpx),\; 
        \zeta + z = \bo y
    \},
$
where $C^s_{\bullet}$ denotes the singular chain complex. 
Intuitively, $\bn_\gf(\zeta)$ is the smallest value of $r$ such that a chain 
homologous to $\zeta$ in $\cpx$ appears in the sub-level-set. This value of 
course depends only on the homology class of $\zeta$. Thus, we have a weight 
function $\bn^*_\gf: \Hom_d(\cpx) \rightarrow \R$.

\begin{lemma}\label{l:funcbottle}
    For any weight function $\wf$ on $d$-simplices of $\cpx$, there is a generic 
    simplex-wise linear function $\gf$ on the barycentric subdivision of $\cpx$, 
    such that for any homology class 
    $h \in H_{d}(\cpx)$, $\bn^*_\gf(h') = \bn^*_\wf(h)$, where $h'$ is the image of $h$ in the subdivision.
\end{lemma}


\begin{proof}
    Let $\cpx'$ denote the subdivision of $\cpx$. Recall that for each simplex 
    $\sigma$ of $\cpx$ there is a vertex $v(\sigma)$ in $\cpx'$. If $\sigma$ is 
    a $d$-simplex, we set $\gf(v(\sigma)) = \wf(\sigma)$. 
    For all other vertices $v$ of $\cpx'$ we define $\bn_\gf(v)$ to be a very 
    small positive number. We then replace these weights with positive integers 
    while maintaining their order. It is easy to check that our function 
    satisfies the statement of the lemma.
\end{proof}

Note that we use the barycentric subdivision simply to give a finer level of granularity on the sub-level sets. This subdivision appears to be necessary for the construction of the function $\gf$. 


\subsection{Bottleneck weight function and persistent homology}
\label{ssec:persistentSolution}

A homology class $h \in H(\cpx)$ is a set of cycles such that the difference 
of any two of the cycles is a boundary chain. Homology classes are intuitively 
referred to as homological features. Persistent homology tries to measure the 
importance of these features. For details see~\cite{EdHa2010}. 

Let the set of simplices of $\cpx$ be ordered such that for each simplex $\spx$, 
the simplices on the boundary of $\spx$ appear before $\spx$ in the ordering. For 
instance, this ordering can be given by the time that a simplex is added, if we 
are building the complex $\cpx$ by adding a simplex at a time. Of course, we need 
the boundary of a simplex to be present before adding it. Let 
\[ 
    \es = \cpx_0 \subset \cpx_1 \cdots \subset \cpx_{n-1} \subset \cpx_{n} = \cpx 
\] 
be the sequence of complexes such that $\cpx_i$ consists of the first $i$ 
simplices in the ordering. Such a sequence is called a \emph{filtration}. 
For $j \geq i$, let $f^{i,j}: \cpx_i \subset \cpx_j$ be the inclusion and 
$f^{i,j}_\sharp$ the induced homomorphism on the chain groups. The homology 
groups $H_d(\cpx_i)$ change as we add simplices. We want to track homology 
features during these additions.

For $0 \leq i \leq j \leq n$, the $d$-dimensional 
\emph{persistent homology group} $H^{i,j}_d$ is the quotient
\[
    H_d^{i,j} = \frac{
            f^{i,j}_\sharp(Z_d(\cpx_i))
        }{
            B_d(\cpx_j) \cap f^{i,j}_\sharp(Z_d(\cpx_i))
        }.
\]
In words, this is the group of those homology classes of $H_d(\cpx_j)$ 
which contain cycles already existing in $\cpx_i$.

We give now an alternate description of the persistent homology classes.
The cycles representing homology features allow us to relate the classes of 
different spaces to each other. We will consider, in each $\cpx_i$, a basis of 
homology and assign to each homology class in these bases a cycle which we call 
a \emph{p-representative} cycle. Consider $\cpx_i$ and let $\spx$ be a 
$d$-simplex such that $\cpx_i \cup \{\spx\} = \cpx_{i+1}$. There are two 
possibilities for the change that adding $\spx$ causes in the homology groups 
of $\cpx_i$. 
\begin{enumerate}
    \item $[\bo_d(\spx)] = 0$ in $\cpx_i$. This implies there is a $d$-chain $b$ such that 
    $\bo_d(b) = \bo_d(\spx)$. Therefore, $\bo_d(b+\spx)=0$. It is easily seen 
    that the cycle $z = b + \spx$ is not a boundary in $\cpx_{i+1}$. We say that the cycle 
    $b + \spx$ and the class $h = [b + \spx]$ are \emph{born} at time $i + 1$ 
    or at $\cpx_{i+1}$. It follows that $H_k(\cpx_{i+1}) = H_k(\cpx_i)$ for 
    $k \neq d$ and $H_d(\cpx_{i+1}) = H_d(\cpx_{i})\oplus([z])$, where $(x)$ 
    means the $\Z_2$-vector space generated by $x$. We take the cycle $z$ to be 
    the p-representative for the class $h$ in $\cpx_{i+1}$. Moreover, If $z'$ is a 
    (inductively defined) p-representative for a homology class of $\cpx_i$ we 
    transfer it to be the p-representative of its class in $\cpx_{i+1}$.
    
    \item $[\bo_d(\spx)] \neq 0$ in $\cpx_i$. In this case, adding the simplex $\spx$ causes 
    the class $h = [\bo(\spx)]$ to become trivial. In other words, each 
    $z \in [\bo(\spx)]$ is now a boundary and this class is merged with the 
    class $0$. 
    Since the p-representatives form a basis of homology, $h$ can 
    be written as a summation of these. The Elder Rule tells us that we declare 
    that the youngest p-representative in this representation \emph{dies} entering 
    $\cpx_{i+1}$. Any other class still can be written as summation of existing 
    p-representatives. Note that each p-representative now represents a possibly 
    larger class.
\end{enumerate}

For $0 \leq i \leq j \leq n$, the $d$-dimensional persistent homology group 
$H^{i,j}_d$ consists of the classes, in $H_d(\cpx_j)$, of those $d$-dimensional 
p-representatives which are born at or before $\cpx_i$. Therefore, the 
p-representatives persist through the filtration. At any $i$, they form a basis 
of the homology groups of $\cpx_i$, and their lifetime can be depicted using 
\emph{barcodes}. The \emph{persistence diagram} encodes the birth and death 
indices of p-representatives. Note that the non-trivial homology classes of 
$\cpx$ are born at some index but never die. From the above explanation the 
following can be observed. We omit the proof.

\begin{proposition}
    Let $h \in H_d(\cpx_i)$ be a homology class and assume $h = \sum t_j [z_j]$ 
    where $t_j \in \Z_2$ and the $z_j$ are p-representatives. Then 
    $\sum t_j z_j$ is a bottleneck optimal cycle for $h$ (with respect to the ordering giving rise to the filtration).
\end{proposition}

Notice that there is a choice of $b$ in the first case of the case analysis above.
In general, the p-representatives are not lex-optimal cycles. However, if we 
choose $b$ to be lex-optimal the p-representatives form a lex-optimal basis. 
This set of basis elements can be computed using the persistent homology boundary 
matrix reduction algorithm~\cite{EdHa2010}, as shown in~\cite{cohensteiner2020}. 
This algorithm runs in $O(m^2\ell)$ time where $\ell$ is the number of 
$d$-simplices and $m$ is the number of $d+1$ simplices.
Also using this basis, a lex-optimal cycle can be computed in any given class in 
$O(\ell^2)$ time~\cite{cohensteiner2020}. Of course, these algorithms also compute 
a bottleneck optimal cycle for any given homology class.

\section{An efficient algorithm for 2-dimensional manifolds}
\label{s:algorithm}

In this section, we present a simple algorithm that, given a combinatorial 
2-manifold $\cpx$, weights on the edges, and a 1-dimensional homology class, 
computes a lex-optimal representative cycle in the given class. For simplicity, 
we consider only orientable manifolds without boundary.

Our input $\cpx$ is an edge-weighted, orientable combinatorial 2-manifold, 
therefore, $S := |\cpx|$ is an orientable surface without boundary. Let $m$ be 
the complexity of $\cpx$. Let $z$ be an input cycle on the 1-skeleton. 
Note that if we want an input cycle $z^S$ in $S$ and not in $\cpx$, i.e. the 
cycle is not on the 1-skeleton, then we can compute an homologous cycle $z$ 
on the 1-skeleton with less than $m$ edges in $O(\ell)$ time, where $\ell$ is 
the number of intersections of $z^S$ with the edges of $\cpx$. 

We first construct a minimum spanning tree $T$ of the 1-skeleton of $\cpx$ 
with respect to the given weights. Let $G$ be the dual graph of the 1-skeleton 
of $\cpx$. The weight of an edge in $G$ is equal to the weight of its 
corresponding dual in $\cpx$. Let $Q$ be the maximum spanning co-tree of $\cpx$ 
in $G$ and let $Q^*$ be the edges of $\cpx$ whose duals are in $Q$. As shown 
in~\cite[Lemma 1]{eppstein2003} $T$ and $Q^*$ are disjoint. Let $L$ be the 
edges that are not in $T$ nor in $Q^*$, and recall that the triple $(T,Q,L)$ 
determines a polygonal schema $P$ of $4g$ sides for $S$, where $g$ is the genus 
of the surface. See Figure~\ref{fig:polygon} as example. This means that if we
cut the surface at $T \cup L$ we obtain a disk $D$, and there is an 
identification map $g: D \rightarrow S$ which will ``re-glue'' the disk into a 
surface. Each edge of $T \cup L$ appears twice around the disk, and each edge 
of $Q^*$ is a diagonal of this disk, connecting two vertices of the disk. The 
cutting of the edges of $T \cup L$ and computing the disk $D$ can be done in 
linear time. The two vertices of the disk that the edges in $Q^*$ connect can 
also be computed in linear time, using previous work on computing the minimal 
homotopic paths~\cite{DeyGuha99,Erickson2013}.

\begin{figure}
    \centering
    \includegraphics[width=0.33\textwidth]{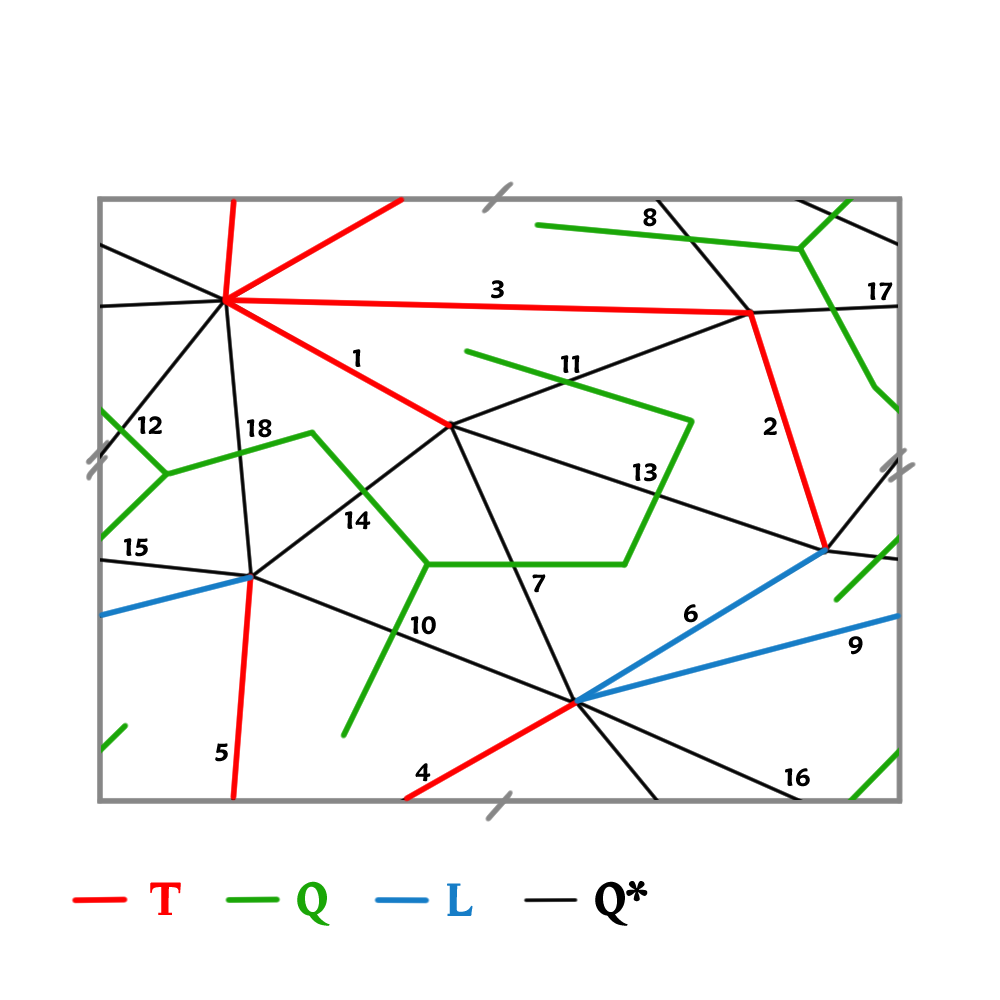}
    \includegraphics[width=0.33\textwidth]{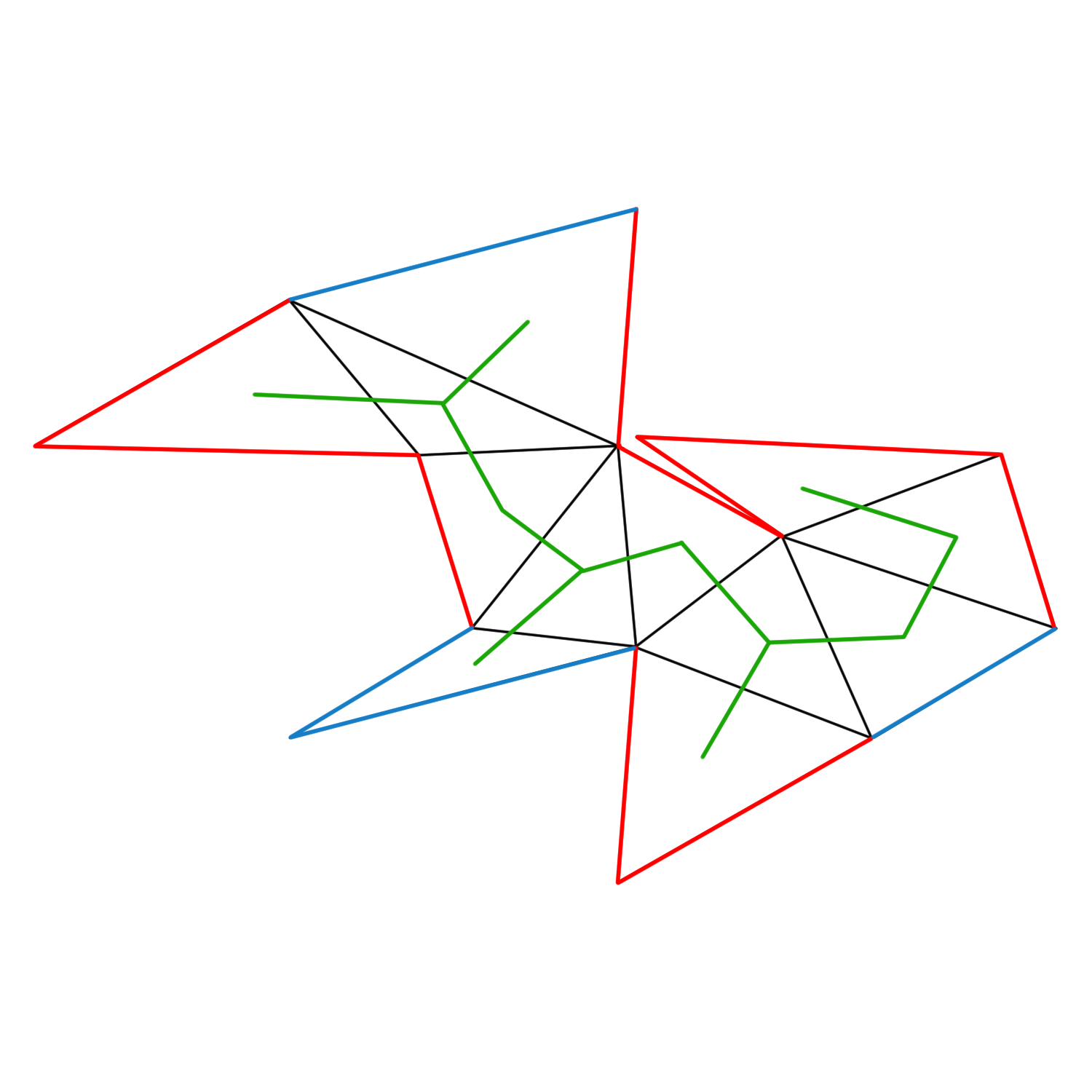}
    \includegraphics[width=0.3\textwidth]{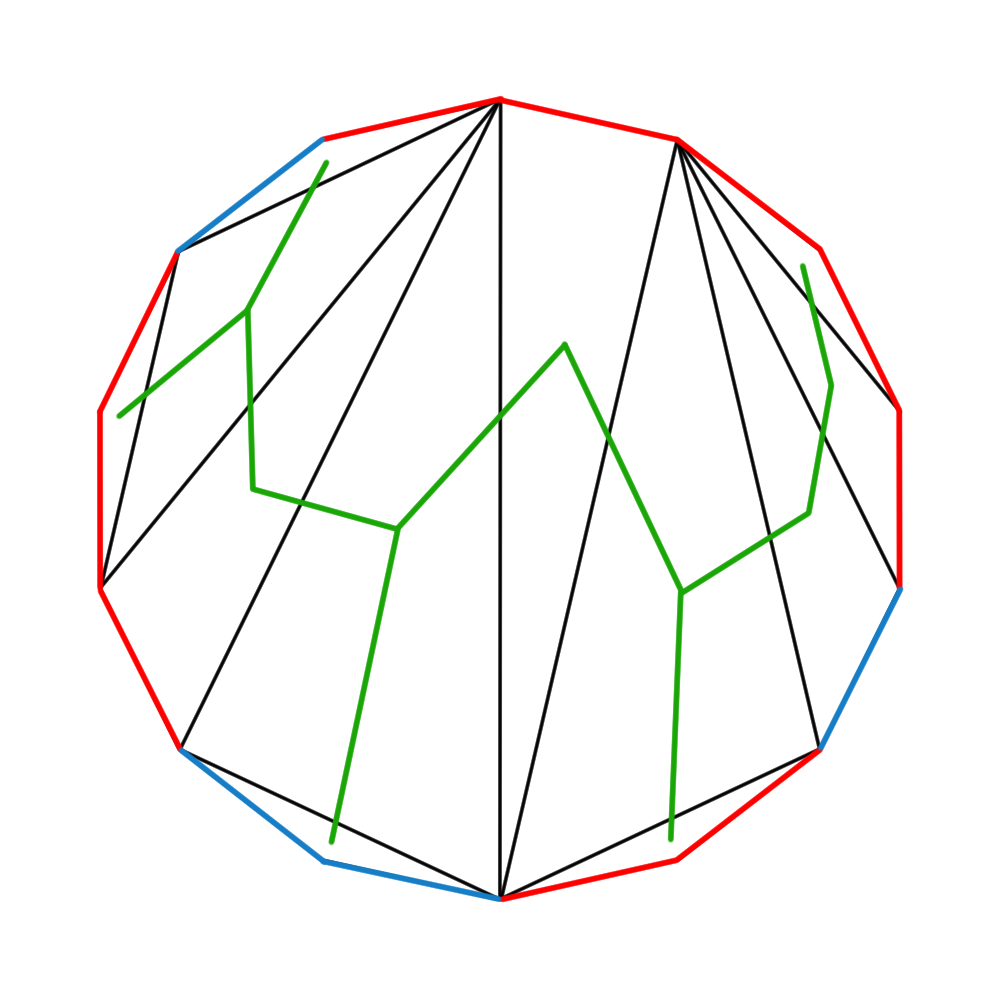}
    \caption{From left to right: a tree co-tree decomposition of a weighted, 
    triangulated torus (with spanning tree shown in red, co-tree shown in green, 
    and the set $L$ in blue); the same torus cut along $T \cup L$; a redrawing 
    of the resulting polygon.}
    \label{fig:polygon}
\end{figure}

During our algorithm, we maintain a data structure $\mathcal{Z}$ which stores a 
circular list of elements of $\Z_2$. The circular list contains a node for each 
boundary edge of the disk $D$.
Note that any edge of $T \cup L$ corresponds to two edges on the boundary of $D$ 
and thus two nodes of $\mathcal{Z}$. 


\subparagraph{Algorithm} 
We compute the lex-optimal cycle $z^*$ in the homology class of the input 
cycle $z$ as follows: 
We start with every node of $\mathcal{Z}$ at value~$0$. Then, for every edge $e$ in 
$z$ in $T \cup L$, we set one of the two nodes corresponding to $e$ to~$1$ and 
keep the other one at~$0$. Finally, for all remaining edges $e$ in $z$, which 
therefore are in $Q^*$, let $a(e)$ be one vertex and $b(e)$ be the other vertex 
which $e$ connects in $D$. We add~$1$ to any node whose corresponding edge of 
$\partial D$ is between $a(e)$ and $b(e)$ in clockwise order. 
At the end, we define the cycle $z^*$ to be the cycle consisting of edges whose 
two corresponding nodes in $\mathcal{Z}$ sum to~$1$.


\subparagraph{Implementation of $\mathcal{Z}$}
The data structure $\mathcal{Z}$ has a single modifying operations: adding a value~$1$ to any
node between two given nodes (inclusive) in clockwise order. In brief, to get a constant-per-operation run-time we accumulate the operations and update the data structure in a single pass. We give now more detail. $\mathcal{Z}$ consists of an array $A$, whose cells are denoted by `nodes' to 
avoid any confusion with complex cells. Each node represents an edge of the 
boundary $\partial D$ of $D$ in the right order.
For any edge $e$ of $z$ in $Q^*$, let $a'(e)$ be the first edge on the 
clockwise path between $a(e)$ and $b(e)$ in $\partial D$ and $b'(e)$ the last 
edge. Additionally to a 0-1 value, each node $c$ in $A$ stores two values 
$s(c)$ and $f(c)$, where $s(c)$ resp. $f(c)$ is the number of edges $e$ of 
$z$ in $Q^*$ whose $a'(e)$ resp. $b'(e)$ corresponds to $c$. For each $e$, the 
cost of updating these two numbers is constant. The final cycle can be computed 
by first computing the value of the first node $A[0]$ and then walking along 
$A$ and updating the value as $A[i] = A[i-1] + s(A[i]) - f(A[i-1])$.


\subparagraph{Correctness}
Let $L = \{\ell_1, \ldots, \ell_{2g}\}$, where the $\ell_i$'s are sorted by 
increasing weight. 
Each edge $\ell_i$ defines a unique cycle when added to the tree $T$, let these 
cycles be denoted by $\Lambda = \{\lambda_1, \ldots, \lambda_{2g}\}$. 
The following lemma is the key to our algorithm's correctness.

\begin{restatable}{lemma}{qedges}
\label{l:qedges}
    Let $q \in Q^*$.  Then there is a 1-chain $c <_L q$ in $T \cup L$, and a 
    2-chain $d$ such that $\bo d = q + c$.
\end{restatable}

\begin{proof}
    The union of the edges in $T$ and $L$ form a cut graph $G$ of the surface, 
    in the sense that the closure of $S-|G|$ is a topological disk $D$. Every 
    edge of $T \cup L$ appears twice on the boundary of $D$, and any $q \in Q^*$ 
    is a diagonal in the polygon $D$. Let $p_1$ and $p_2$ be the two arcs such 
    that $\bo D = p_1 \cup p_2$ and the endpoints of $p_1$ and $p_2$ coincide 
    with those of $q$. Let $\tilde{p}_i \in C_1(\cpx)$ be the 1-chain 
    corresponding to $p_i$, $i = 1,2$, i.e., $\tilde{p}_i = g_\sharp(p_i)$. 
    Recall that $g_\sharp$ is the induced map on chain groups.
    Let $d_1$ be the 2-chain bounded by $p_1$ and $q$ and let 
    $\tilde{d}_1 = g_\sharp(d_1)$. We have $\bo \tilde{d} = q + \tilde{p}_1$, 
    where by $q$ we denote this edge in $D$ and $S$.
    We now claim that every edge in $\tilde{p}_1$ is smaller than $q$. Note that 
    $\tilde{p}_1$ is a chain of $T \cup L$. We consider two cases. First, assume 
    $\tilde{p}_1$ consists only of edges of $T$. In this case, it equals the 
    unique path in $T$ defined by the endpoints of $q$. Since $T$ is a minimum 
    spanning tree our claim is proved.

    Second, assume that $\tilde{p}_1$ is not entirely in $T$. In this case we 
    argue as follows. Let $\ell \in L$ and let $\ell_1$ and $\ell_2$ be the two 
    copies of $\ell$ on $\bo D$. We claim that if $\ell_1$ and $\ell_2$ are on 
    both of the arcs $p_1$ and $p_2$ (that is, if $\ell_1 \in p_1$ and 
    $\ell_2 \in p_2$ or $\ell_2 \in p_1$ and $\ell_1 \in p_2$) then $\ell < q$. 
    Assume for the sake of contradiction that $\ell > q$. Under these conditions, 
    if we remove the dual of $q$ from $Q$ and add the dual of $\ell$, we have 
    reconnected the spanning co-tree split by removing $q$ (since the effect of 
    removing $q$ from $Q$ is adding it to $L$ and thus cutting the disk $D$ at 
    $q$ while the effect of adding $\ell$ to $Q-\{q\}$ is merging the resulting 
    disks at $\ell_1$ and $\ell_2$, thus again forming a single disk). Thus we 
    have increased the weight of the spanning co-tree which is not possible. 
    Therefore, $\ell < q$ or the two copies of $\ell$ appear on one of $p_1$ or 
    $p_2$. It follows that every edge of $L$ (which appears once) in 
    $\tilde{p}_1$ is smaller than $q$ (since appearing twice cancels an edge).
    To finish the proof in this case, we 
    claim that for every edge $t \in T \cap \tilde{p}_1$ there is an edge 
    $\ell \in L \cap \tilde{p}_1$ such that $t < \ell$. It then follows that 
    $\tilde{p}_1 < q$. 

    To prove the claim we argue as follows. $T \cup L$ is a graph on the 
    1-skeleton of $K$ and it is standard and easy to show that any homology 
    class $0 \neq h \in H_1(K)$ contains exactly one cycle of $T \cup L$. Let 
    $z = q + x$ be the cycle formed by adding $q$ to $T$, where $x$ is the path 
    on $T$. We have 
    $z + \bo{\tilde{d}_1} = z + \tilde{p}_1 + q = \tilde{p}_1 + x =: z'$ and 
    $z'$ is a non-empty cycle in $T \cup L$ (If $z'$ were empty then 
    $\tilde{p}_1 = x$, this is not possible since $x$ is in $T$ and 
    $\tilde{p}_1$ is not in $T$ by assumption). Thus $z'$ can also be written 
    as a non-empty summation of the $\lambda_i$. Each $\lambda_i$ has the 
    property that its unique edge $\ell_i \in L$ is larger than its edges in 
    $T$. Since these $\ell_i$ are never cancelled, it follows that for each 
    edge $t \in z' \cap T$ there is an edge $\ell \in z' \cap L$ such that 
    $t < \ell$. Since the $L$-edges are in $\tilde{p}_1$ and not in $x$ our 
    claim is proved.
\end{proof}


With some abuse of notation we also denote the chain on $\bo D$ defined by the 
nodes of $\mathcal{Z}$ with value 1 by $\mathcal{Z}$. Note that in the beginning 
of the algorithm $g_\sharp({\mathcal{Z}}) = z$. The algorithm then repeatedly 
updates $\mathcal{Z}$ by adding the chain $p_1$ returned by the above lemma to 
$\mathcal{Z}$ and adding the chain $\tilde{p}_1 = g_\sharp(p_1)$ to $z$. It updates $\mathcal{Z}$ such that at any time 
$g_\sharp({\mathcal{Z}}) = z$. 
To finish the proof of correctness, it remains to show that the final cycle, 
namely the unique cycle $z_*$ of $T \cup L$ in the class of $z$, is indeed the 
lex-min cycle. To see this, assume on the contrary that there is a cycle $y$ such 
that $[y] = [z]$ and $y < z_*$. Since there is a unique cycle of any class in 
$T \cup L$, $y$ has to contain an edge of $Q^*$ hence can be made smaller, which 
contradicts minimality of $y$. Therefore, the cycles of $T \cup L$ are indeed 
the lex-min representatives of homology classes.

\begin{restatable}{theorem}{lexoptalg}
    Let $\cpx$ be a simplicial complex which is a closed orientable combinatorial 
    2-manifold and let $m$ be its number of simplices. There is an algorithm that 
    computes a lex-optimal basis for the 1-dimensional homology of $\cpx$ in 
    $O(m \log(m))$ time. Moreover, we can compute a 
    lex-optimal representative for any given 1-homology class within the same run-time.
\end{restatable}

\begin{proof}
    We have proved that the algorithm correctly computes the lex-optimal cycle 
    homologous to $z$. We show that the basis $\Lambda$ is lex-optimal basis. 
    First note that by Lemma~\ref{l:qedges} every non-trivial cycle $y$ is 
    homologous to a cycle $y' \leq_L y$ such that $y'$ is a subset of 
    $T \cup L$. Since these must contain some $\ell_i$, it follows that the 
    smallest non-trivial cycle contains only $\ell_1$ and edges of $T$, and 
    hence is $\lambda_1$.
    
    Assume inductively that $\Lambda_i = \{ \lambda_1, \ldots, \lambda_i \}$ is 
    a lex-optimal basis for the vector space $(\Lambda_i) \subset H_1(\cpx)$. 
    We claim that $\lambda_{i+1}$ is the smallest cycle in classes in the set $H_1-(\Lambda_i)$. Consider any non-trivial cycle $y$ and decrease it to $y'$ as above. 
    To see that $\lambda_{i+1}$  is the smallest cycle,
    note that $y' \cap L$ must contain some $l_j$ 
    larger than $\lambda_i$, since otherwise $[y]\in (\Lambda_i)$; the 
    smallest cycle with this property is $\lambda_{i+1}$.
    
    Constructing $T$, the dual graph and $Q$ takes at most $O(m \log m)$ time.
    Since we perform one update operation on $\mathcal{Z}$ per edges of $Q$ 
    the total running time is $O(m \log m)$.
    
    \end{proof}

\section{Reductions}\label{s:reductions}

In this section, we first reduce solving a system of linear equation $Ax = b$, 
with $A$ sparse, to computing the bottleneck-optimal homologous cycle 
problem for a 3-manifold given as a subset of the Euclidean 3-space. We 
then use this reduction to deduce hardness results for similar homological 
computations for 3-manifolds and 2-complexes in 3-space. Due to space 
constraints, some proofs of this section can only be found in 
the full version of the paper.

Let $A = (a_{ij})$, $i,j \in \{1,\ldots,n\}$, be an $n \times n$ square matrix 
with values in $\Z_2$. 
Let $A_i$ denote the $i$-th column, and $A^t_i$ denote 
the $i$-the row of $A$. Let $x = (x_1, \ldots, x_n)$ be the vector of the $n$ 
variables of the system $Ax = b$, and $b = (b_1, \ldots, b_n)$.

\begin{figure}
    \centering
    \includegraphics[scale=0.65]{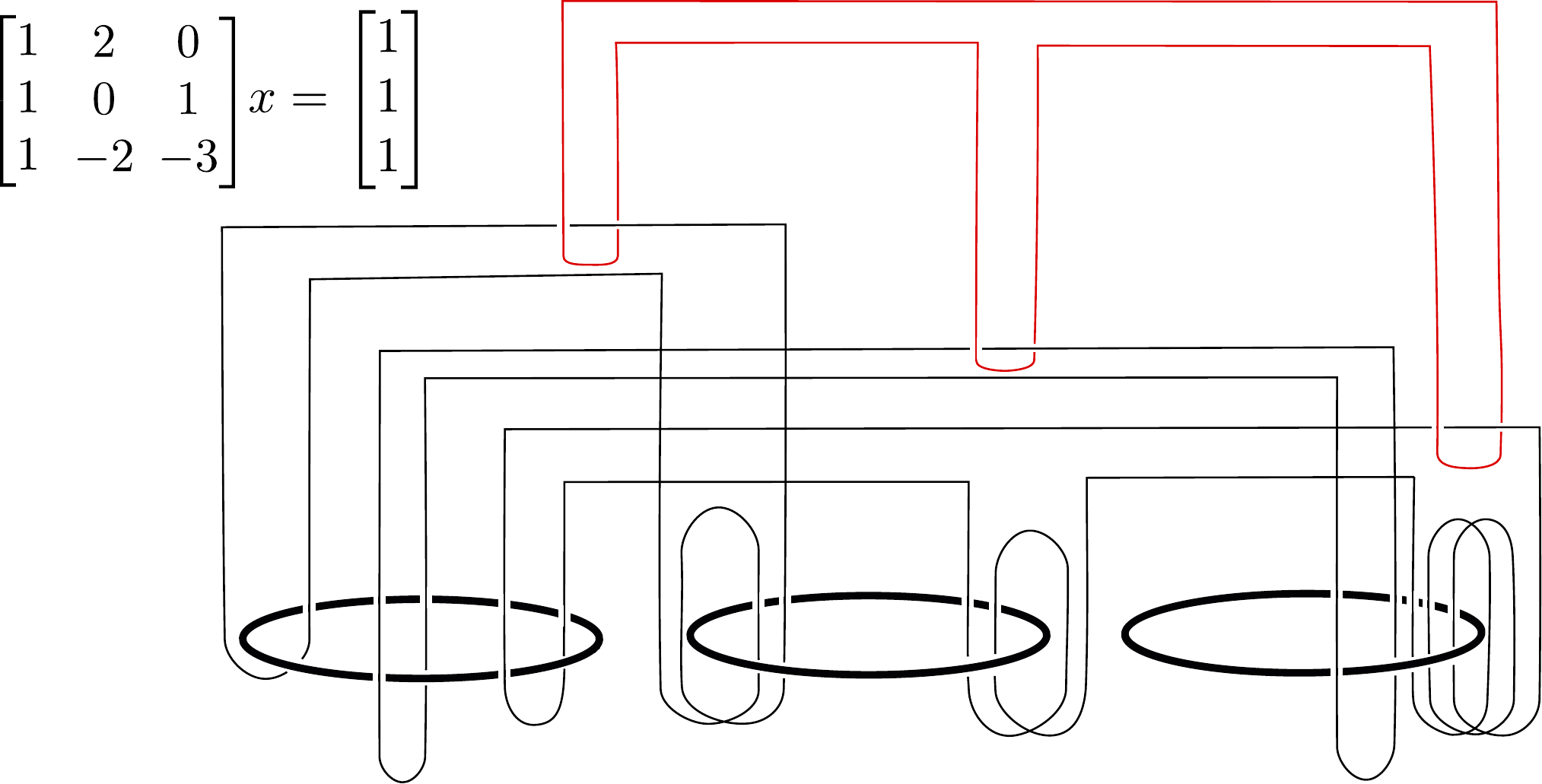}
    \caption {This figure is a link diagram for the reduction of the indicated 
    linear system. The thick round circles are the $\lambda_i$, and the link 
    components are drawn in thin black and red. For all of the crossings between 
    $\Ld_i$, the vertical strand goes over the horizontal strand. In other words, 
    these are not linked with each other. Note that in this example we are using 
    integer matrix and integer linking number to showcase how the reduction works 
    for integers. In the text we are concerned with 0-1 matrices. One also sees 
    here how the need for generating large linking numbers (in absolute value) 
    increases the complexity of the reduction.}
    \label{fig:reduction}
\end{figure}

From the given system $Ax = b$, we first construct a link diagram $\Ld'$.
We start by drawing $n$ round circles in the plane, whose collection we denote 
by $\Lambda' = \{\lambda'_1, \ldots, \lambda'_n \}$; see the thick circles in 
the Figure~\ref{fig:reduction} for an illustration. For each row $A^t_i$ of $A$, 
we draw a component of the link $\Ld'$, denoted $L'_i$, such that its linking 
number is non-zero with $\lambda'_j$ if and only if $a_{ij} = 1$; this can be 
accomplished simply by linking $L_i'$ appropriately with $\lambda'_i$ depending 
on the value of $a_{ij}$. As we wish the $L_i$'s to not link with each other, 
any crossings between a fixed $L_i$ and $L_j$ are simply set to be all over 
(or all under), so that they will remain unlinked. Again, we refer to 
Figure~\ref{fig:reduction}, where example knots $L'_1, L_2'$, and $L_3'$ are 
depicted by thin black lines. We add one final knot, which we denote as 
$\zeta'$, to the link $\Ld'$ so that  its linking number with $\lambda_i'$ is 
non-zero if and only if $b_i = 1$; this can be accomplished by linking $\zeta'$ 
once with each $L_i'$. See the top knot shown in red in 
Figure~\ref{fig:reduction} for an illustration.

\begin{restatable}{lemma}{linkcrossings}
    Let $A$ be such that each row of $A$ has at most $c$ non-zero entries. 
    Then the link diagram $\Ld'$ has $O(cn^2)$ crossings.
\end{restatable}

In the next step, we construct a spatial link $\Ld$ from the link diagram 
$\Ld'$, such that the knots appear in the 1-skeleton of a triangulation of a 
3-ball. This is standard and can be done in $O(m \log(m))$-time where $m$ is 
the number of crossings of the link diagram $\Ld'$~\cite[Lemma~7.1]{HLP99}. 
The resulted space has complexity $O(m)$. Our diagram has $O(n^2)$ many 
crossings, therefore this construction takes $O(n^2 \log{n})$ time and we 
obtain triangulation of a ball with complexity $O(n^2)$. The spatial link $\Ld$ 
corresponding to $\Ld'$ is a set of disjoint simple closed curves in the 
1-skeleton of a triangulation of a 3-ball $B^3$. We denote the spatial knots 
corresponding to $L'_i$ by $L_i$, and analogously we name other components of 
$\Ld$. 

Consider the sub-link $\mathcal{N}$ of $\Ld$ consisting of the components $L_i$. 
We define the manifold $\mf$ to be the link-complement of the link 
$\mathcal{N}$. This link-complement, by definition, is obtained by removing the 
interior of a thin polyhedral tubular neighborhood 
of each component of $\mathcal{N}$. This construction is again 
standard, and a triangulation of $\mf$ can be constructed in linear time form the 
spatial link~\cite{HLP99}. Therefore, the 3-manifold $\mf$ is a subset of a 3-ball 
$B^3$, and has $n$ boundary components. By extra subdivisions, if necessary, 
we can make sure that in the interior of $\mf$, $\zeta$, $\lambda_i$ and $b$ are 
simple, disjoint, closed curves in the 1-skeleton. To do this, it is enough to 
make sure this property holds in every tetrahedron.

The cycle $\zeta$ is the input cycle in our instance of the bottleneck-optimal 
homologous cycle problem. We still need to define our edge weights, which will 
be based on an ordering of the edges of $\mf$. Let $\{e_{ij}\}$ be the set of 
edges in the cycle $\lambda_i$. First, we make sure that every edge not in some 
$\lambda_i$ is larger than any $e_{ij}$. Second, if $i < i'$, we make sure that, 
for all $j,j'$, $e_{ij}$ is smaller than $e_{i'j'}$. This finishes construction
of our problem instance.

Let $\mu_i$, $i \in \{ 1, \ldots, n \}$, be meridians of the knots $L_i$ in 
$\mf$. This is a circle on the boundary component of $\mf$ corresponding to 
$L_i$. It is well-known that the homology group $H_1(\mf)$ is a 
$\Z_2$-vector space  with the basis $\{[\mu_1], \ldots, [\mu_n]\}$ isomorphic 
to $\Z_2^{n}$. 

\begin{restatable}{lemma}{reduction}
\label{l:reduction}
    The following hold:
    \begin{enumerate}
        \item If there is a vector $x \in \Z_2^n$ such that $Ax = b$ then any 
        bottleneck-optimal cycle in the class $[\zeta]$ is a summation of the 
        cycles $\lambda_j$.
    
        \item If there exists a bottleneck-optimal cycle $z_*$ in the class 
        $[\zeta]$ such that $z_*=\sum_{i=1}^n x_i \lambda_i$ then the vector 
        $x = (x_1, \ldots, x_n)^t$ is a solution to $Ax = b$.
    \end{enumerate}
\end{restatable}

\begin{proof}
    First, observe that for the class $[\lambda_j]$ we have 
    $[\lambda_i] = \sum_{j=1}^n a_{ji} [\mu_j]$. The left of 
    Figure~\ref{fig:meridian} depicts a 2-chain realizing this relation. 
    If we map the basis element $[\mu_j]$ to the $j$-th standard basis element 
    $e_j$, then we have defined an isomorphism $H_1(M)\cong \Z_2^n$ in which 
    the class $[\lambda_i]$ maps to the column $A_i$. Second, note that, with 
    a similar argument, $[\zeta] = \sum_{i=1}^n b_j [\mu_j]$, see 
    Figure~\ref{fig:meridian} right. Thus $[\zeta]$ maps to $b$ under the 
    isomorphism. It follows that $Ax = b$ if and only if 
    $\sum x_j [\lambda_j] = [\zeta]$. The second statement follows.
    
    If $x$ is a solution to $Ax=b$, then the cycle $z = \sum x_j \lambda_j$ 
    belongs to the class $[\zeta]$.  Any cycle which is not entirely a subset 
    of the edges of the $\lambda_i$'s, and hence a summation of the 
    $\lambda_i$, contains some edge which is larger than all the edges of the 
    $\lambda_i$'s and therefore has weight more than $z$. It follows that any 
    bottleneck-optimal cycle is a summation of the $\lambda_i$ or, a subset of 
    them, since these are disjoint simple cycles. This proves the first 
    statement.
\end{proof}

\begin{figure}
    \centering
    \includegraphics[width=0.35\textwidth]{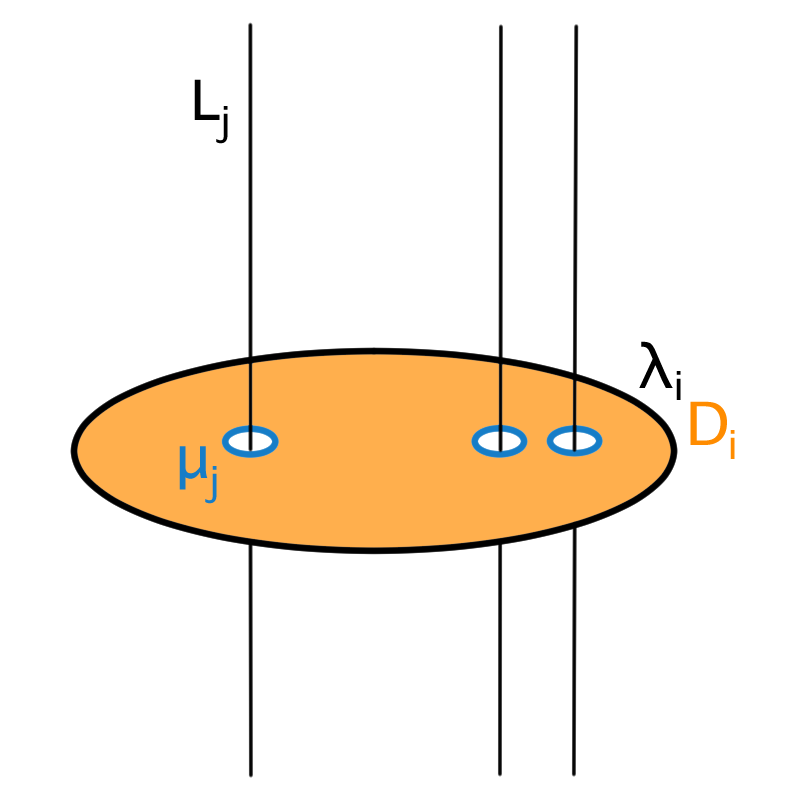}
    \includegraphics[width=0.35\textwidth]{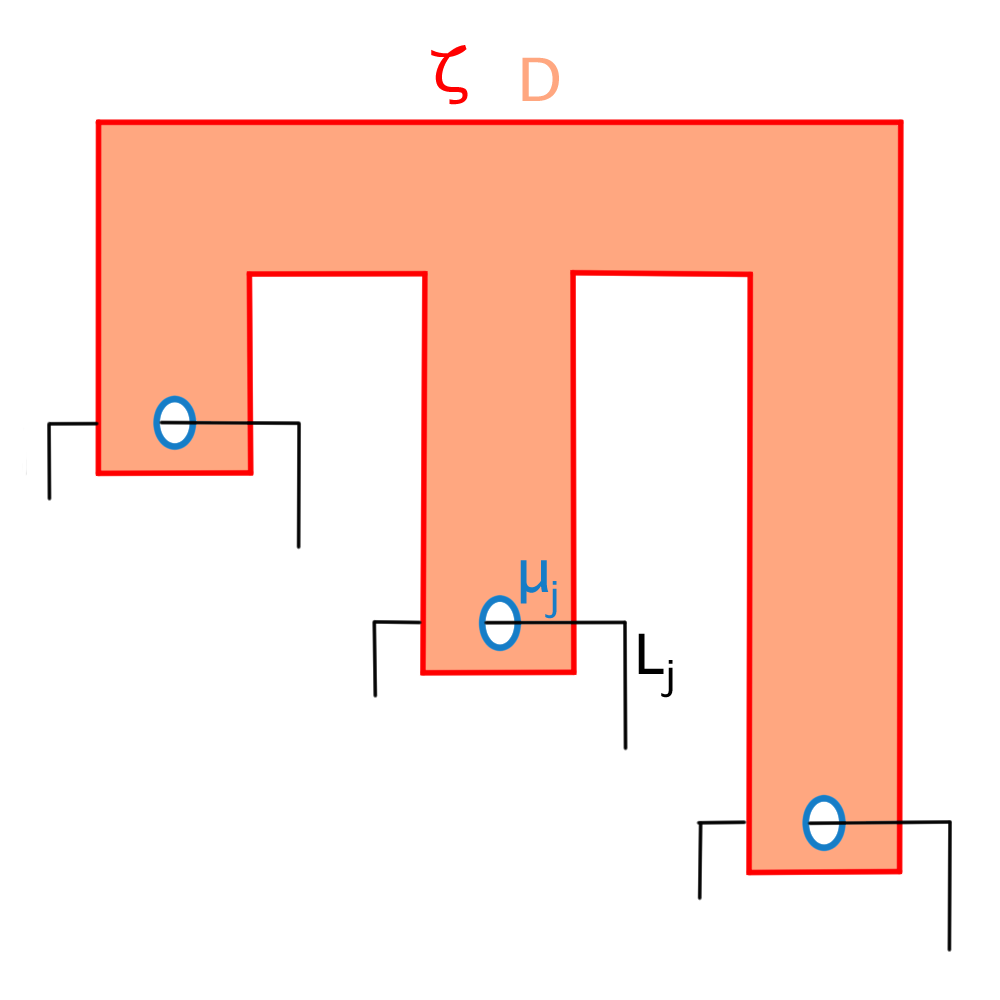}
    \caption{$\bo D_i = \lambda_i + \sum_j a_{ji} \mu_j$ (left), 
    $\bo D = \zeta + \sum_j b_j \mu_j$ (right). }
    \label{fig:meridian}
\end{figure}

\begin{restatable}{theorem}{sparsematrix}
    Solving the system of equations $Ax = b$ where $A$ is a sparse $\Z_2$-matrix 
    reduces in $O(n^2 \log n)$ time to the bottleneck-optimal homologous cycle 
    problem with $\Z_2$-coefficients for a 3-manifold of size $O(n^2)$ given 
    as a subset of $\R^3$.
\end{restatable}

\begin{proof}
    Given the system $Ax=b$ we have already constructed our instance. If the 
    bottleneck-optimal cycle $z_*$ returned by any algorithm that solves the 
    Bottleneck-OHCP problem uses only edges in $\bigcup \lambda_i$, then the 
    second statement of Lemma~\ref{l:reduction} implies that we can find a 
    solution by determining which $\lambda_i$ appear in $z_*$. This can be done 
    in linear time. On the other hand, if $z_*$ uses some edge not in 
    $\bigcup \lambda_i$, then there is no solution to the system by the first 
    statement of Lemma~\ref{l:reduction}.
\end{proof}

Although we have not defined the integer homology groups, it is almost 
immediate that the above reduction works also with $\Z$-coefficients.

\begin{corollary}
    The (1-dimensional) lex-optimal homologous cycle problem for 3-manifolds in $\R^3$ of size $n^2$ cannot be solved more efficiently than the time required to solve a 
    system of equations $Ax = b$ with $A$ a sparse $n \times n$ matrix, if the 
    latter time is $\Omega(n^2 \log(n))$.
\end{corollary}

As noted before, the persistent boundary reduction algorithm can compute a 
lex-optimal cycle in $O(lm^2)$ time~\cite{cohensteiner2020}, where $m$ is the 
number of $d+1$-simplices and $l$ is the number of $d$-simplices. Although a set 
of persistent generators can be computed in matrix multiplication 
time~$\cite{MMS11}$, we do not know that the lex-optimal cycle can be found in 
matrix multiplication time, as it is unclear if the divide and conquer strategy 
from $\cite{MMS11}$ would work on our problem.

\begin{restatable}{corollary}{mmtime}
    A set of sub-level-set persistent homology generators for a 3-manifold $M$ 
    or a 2-complex $\cpx$ of size $n^2$ in $\R^3$ and a generic simplex-wise 
    linear function $f: M \rightarrow \R$ cannot be computed more efficiently 
    than the time required to compute a maximal set of independent columns in 
    an $n \times n$ sparse matrix $A$, if the latter time is 
    $\Omega(n^2 \log(n))$.
\end{restatable}

 \begin{proof}
     Recall that in our reduction above, the cycles $\lambda_i$ correspond to 
     columns of the matrix $A$. After a subdivision, we can define a 
     simplex-wise linear function such that bottleneck weight function equals 
     the sub-level set bottleneck weight function, by Lemma~\ref{l:funcbottle}. 
     With respect to this function, the (subdivision of the) cycles $\lambda_i$ 
     are the first cycles in their homology classes that appear in a sub-level 
     set. Only an independent set of these cycles will remain until the end. 
     These cycles therefore determine a subset of the columns of $A$ which are 
     independent.
    
     For the 2-complex $\cpx$, we can simply take the 2-skeleton of the manifold 
     in the reduction.
     There will be no change in 1-dimensional persistent homology groups.
\end{proof}

As noted in the introduction, the above results are in a strong contrast with 
the results of Dey~\cite{Dey2019}. In other words, if the complex is of size 
$O(n^2)$ and the given function on the simplicial complex $\cpx$ is a height 
function then one can compute the generators in $O(n^2 \log n)$ time, whereas, 
for a general function, one cannot do better than computing a maximal set of 
independent columns for a given sparse matrix $A$ of size $n$. To the best of 
our knowledge, the best deterministic algorithm for this operation takes at 
least $O(n^\omega)$ time, where $\omega$ is the exponent of matrix 
multiplication.

\begin{restatable}{corollary}{threemanifold}
    The persistence diagram for a 2-complex or a 3-manifold of size $n^2$ in 
    $\R^3$ and a generic simplex-wise linear function $f: |\cpx| \rightarrow \R$ 
    cannot be computed more efficiently than the time required to compute the 
    rank of a sparse $n \times n$ matrix $A$,  if the latter time is 
    $\Omega(n^2 \log(n))$. 
\end{restatable}

 \begin{proof}
     The number of 1-dimensional points with infinite death times and which 
     correspond to the $\lambda_i$ is the number of independent columns of $A$. 
     We can determine if a point belongs to a $\lambda_i$ by checking its birth 
     time. The point represents some $\lambda_i$ if and only if the birth time 
     is smaller than the weight of edges not contained in these cycles.   
 \end{proof}

Again the above theorem should be compared with results of \cite{Dey2019}, where 
the persistence is computed in $O(n^2 \log n)$ time for a 2-complex in 3-space 
of size $n^2$ and a height function.


\bibliography{bib}

\end{document}